\documentclass[letterpaper,11pt]{article}

\usepackage[margin=1in]{geometry}

\usepackage[numbers]{natbib}

\usepackage{lmodern}
\usepackage{graphicx}
\usepackage[caption=false,subrefformat=parens,labelformat=parens]{subfig}
\usepackage{amsmath,amsthm,amsfonts,amssymb}
\usepackage{mathtools}
\usepackage{hyperref}

\graphicspath{{fig-ipe/}}

\newtheorem{theorem}{Theorem}
\newtheorem{lemma}{Lemma}
\theoremstyle{remark}
\newtheorem{remark}{Remark}

\newcommand{\calP}{\mathcal{P}}
\newcommand{\Ahat}{\hat{A}}

\newcommand{\Ihat}{\hat{I}}
\newcommand{\Phat}{\hat{P}}
\newcommand{\Xhat}{\hat{X}}
\newcommand{\alphahat}{\hat{\alpha}}
\newcommand{\muhat}{\hat{\mu}}

\DeclareMathOperator{\poly}{poly}
\DeclarePairedDelimiter\abs{\lvert}{\rvert}

\DeclarePairedDelimiter\List{\langle}{\rangle}
\newcommand{\concat}{\cdot}

\newcommand{\DSMCYC}[1]{#1\textsc{-dsm-}\allowbreak\textsc{cyc}}
\newcommand{\DSMICYC}[1]{#1\textsc{-dsmi-}\allowbreak\textsc{cyc}}

\title{On the Existence of Three-Dimensional Stable Matchings \\
with Cyclic Preferences%
\thanks{Department of Computer Science, University of Texas at Austin, 2317
Speedway, Stop D9500, Austin, TX 78712--1757. Email:
\{geocklam, plaxton\}@cs.utexas.edu.}}

\author{Chi-Kit Lam \and C. Gregory Plaxton}

\begin{document}

\maketitle

\begin{abstract}


We study the three-dimensional stable matching problem with cyclic preferences.
This model involves three types of agents, with an equal number of agents of
each type. The types form a~cyclic order such that each agent has a~complete
preference list over the agents of the next type. We consider the open problem
of the existence of three-dimensional matchings in which no triple of agents
prefer each other to their partners. Such matchings are said to be weakly
stable. We show that contrary to published conjectures, weakly stable
three-dimensional matchings need not exist. Furthermore, we show that it is
NP-complete to determine whether a~weakly stable three-dimensional matchings
exists. We achieve this by reducing from the variant of the problem where
preference lists are allowed to be incomplete. Our results can be generalized
to the $k$-dimensional stable matching problem with cyclic preferences for
$k \geq 3$.

\end{abstract}


\section{Introduction}

The study of stable matchings was started by Gale and Shapley~\cite{GS62}, who
investigated a~market with two types of agents. The two-dimensional stable
matching problem involves an equal number of men and women, each of whom has
a~complete preference list over the agents of the opposite sex. The goal is to
find a~matching between the men and the women such that no man and woman prefer
each other to their partners. Matchings satisfying this property are said to be
stable. Gale and Shapley showed that a~solution for the two-dimensional stable
matching problem always exists and can be computed in polynomial time. Their
result also applies to the variant where preference lists may be incomplete
due to unacceptable partners, and the number of men may be different from
the number of women.

The problem of generalizing stable matchings to markets with three types of
agents was posed by Knuth~\cite{Knu97}. In pursuit of an existence theorem
and an elegant theory analogous to those of the Gale-Shapley model, the
three-dimensional stable matching problem has been studied with respect to
a~number of preference structures. When each agent has preferences over pairs
of agents from the other two types, stable matchings need not
exist~\cite{Alk88,NH91}. Furthermore, it is NP-complete to determine whether
a~stable matching exists~\cite{NH91,Sub94}, even if the preferences are
consistent with product orders~\cite{Hua07}. When two types of agents care
primarily about each other and secondarily about the remaining type, a~stable
matching always exists and can be obtained by computing two-dimensional stable
matchings using the Gale-Shapley algorithm in a~hierarchical
manner~\cite{Dan03}. When the types form a~cyclic order such that each type of
agent cares primarily about the next type and secondarily about the other type,
stable matchings need not exist~\cite{BGJK04}.

A~prominent problem mentioned in several of the aforementioned
papers~\cite{BGJK04,Hua07,NH91} is the three-dimensional stable matching
problem for the case where the types form a~cyclic order such that each type of
agent cares only about the next type and not the other type. Following the
terminology of the survey of Manlove~\cite{Man13}, we call this the
three-dimensional stable matching problem with cyclic preferences
($\DSMCYC{3}$), and refer to the three types of agents as men, women, and dogs.
A~number of stability notions~\cite{Hua07} can be considered in $\DSMCYC{3}$.
In this paper, we focus on weak stability, which is the most permissive one and
has received the most attention in the literature. It is known that determining
whether a~$\DSMCYC{3}$ instance has a~strongly stable matching is
NP-complete~\cite{BM10}. For the variant where ties are allowed, determining
the existence of a~super-stable matching is also NP-complete~\cite{Hua10}.
However, it remained an open problem for weakly stable matchings in
$\DSMCYC{3}$.

In $\DSMCYC{3}$, there are an equal number of men, women, and dogs. Each man
has a~complete preference list over the women, each woman has a~complete
preference list over the dogs, and each dog has a~complete preference list over
the men. A~family is a~triple consisting of a~man, a~woman, and a~dog.
A~matching is a~set of agent-disjoint families. A~family is strongly blocking if
every agent in the family prefers each other to their partners in the matching.
A~matching is weakly stable if it admits no strongly blocking family. This
problem is related to applications such as kidney exchange~\cite{BM10} and
three-sided network services~\cite{CJ13}.

The formulation of $\DSMCYC{3}$ first appeared in the paper of Ng and
Hirsch\-berg~\cite{NH91}, where it is attributed to Knuth. Using a~greedy
approach, Boros et al.~\cite{BGJK04} showed that every $\DSMCYC{3}$ instance
with at most~$3$ agents per type has a~weakly stable matching. Their result
also applies to the $k$-dimensional generalization of the problem, which we
call $\DSMCYC{k}$. For $k \geq 3$, they showed that every $\DSMCYC{k}$ instance
with at most~$k$ agents per type has a~weakly stable matching. Using a case
analysis, Eriksson et al.~\cite{ESS06} showed that every $\DSMCYC{3}$ instance
with at most~$4$ agents per type has a~weakly stable matching. Despite
unbalancedness and the lack of stable effective functions, Eriksson et
al.\ conjectured that every $\DSMCYC{3}$ instance has a~weakly stable matching.
More precisely, they conjectured that for an extension of $\DSMCYC{3}$ with the
strongest link rule, every instance with at least~$2$ agents per type has at
least two weakly stable matchings. Using an efficient greedy procedure,
Hofbauer~\cite{Hof16} showed that for $k \geq 3$, every $\DSMCYC{k}$ instance
with at most~$k+1$ agents per type has a~weakly stable matching. Using
a~satisfiability problem formulation and an extensive computer-assisted search,
Pashkovich and Poirrier~\cite{PP18} showed that every $\DSMCYC{3}$ instance
with exactly~$5$ agents per type has at least two weakly stable matchings.
Escamocher and O'Sullivan~\cite{EO18} showed that the number of weakly stable
matchings is exponential in the size of the $\DSMCYC{3}$ instance if agents of
the same type are restricted to have the same preferences. They also
conjectured that for unrestricted $\DSMCYC{3}$ instances, there are
exponentially many weakly stable matchings.

Hardness results are known for some related problems. For the variant
of $\DSMCYC{3}$ where preference lists are allowed to be incomplete, which we
refer to as $\DSMICYC{3}$, Bir\'{o} and McDermid~\cite{BM10} showed that
determining whether a~weakly stable matching exists is NP-complete.
Farczadi et al.~\cite{FGK16} showed that determining whether a given perfect
two-dimensional matching can be extended to a three-dimensional weakly stable
matching in $\DSMCYC{3}$ is also NP-complete. However, the existence of weakly
stable matchings in $\DSMCYC{3}$ remained unresolved. Manlove~\cite{Man13}
described it as an ``intriguing open problem", and Woeginger~\cite{Woe13}
classified it as ``hard and outstanding".

\subsubsection*{Our Techniques and Contributions}

In this paper, we show that there exists a~$\DSMCYC{3}$ instance that has no
weakly stable matching. This disproves the conjectures of Eriksson et
al.~\cite{ESS06} and Escamocher and O'Sullivan~\cite{EO18}. Furthermore, we
show that determining whether a~$\DSMCYC{3}$ instance has a~weakly stable
matching is NP-complete. We achieve this by reducing from the problem of
determining whether a~$\DSMICYC{3}$ instance has a weakly stable matching. Our
results generalize to $\DSMCYC{k}$ for $k \geq 3$.

Our main technique involves converting each agent in $\DSMICYC{3}$ to a~gadget
consisting of one non-dummy agent and many dummy agents. The dummy agents in
our gadget give rise to chains of admirers. (See Remark~\ref{rem:admirer} in
Section~\ref{sec:kdsmcyc-gadget}.) By applying the weak stability condition to
the chains of admirers, we are able to obtain some control over the partner of
the non-dummy agent.

\subsubsection*{Organization of This Paper}

In Section~\ref{sec:preliminaries}, we present the formal definitions of
$\DSMCYC{k}$ and $\DSMICYC{k}$. In Section~\ref{sec:kdsmicyc}, we show that
the NP-completeness result of Bir\'{o} and McDermid~\cite{BM10} can be extended
to $\DSMICYC{k}$. In Section~\ref{sec:kdsmcyc}, we show that $\DSMCYC{k}$ is
NP-complete by a~reduction from $\DSMICYC{k}$. In Section~\ref{sec:conclusion}, we
conclude by mentioning some potential future work.

\section{Preliminaries}
\label{sec:preliminaries}

In this paper, we use $\List{z \in Z \mid \calP(z)}$ to denote the list of
all tuples $z \in Z$ satisfying predicate $\calP(z)$, where the tuples are
sorted in increasing lexicographical order. Given two lists~$Y$ and~$Z$, we
denote their concatenation as $Y \concat Z$. For any $k \geq 1$, we use
$\oplus_k$ to denote addition modulo~$k$.

\subsection{The Models}

Let $k \geq 2$. The $k$-dimensional stable matching problem with incomplete
lists and cyclic preferences ($\DSMICYC{k}$) involves a~finite set $A =
I \times \{0, \dots, k-1\}$ of agents, where each agent
$\alpha = (i, t) \in A$ is associated with an identifier~$i$ and a~type~$t$.
(When $k = 3$, we can think of the sets $I \times \{0\}$, $I \times \{1\}$,
and $I \times \{2\}$ as the sets of men, women, and dogs, respectively.) Each
agent~$\alpha = (i, t) \in A$ has a~strict preference list~$P_\alpha$ over
a~subset of agents of type~$t' = t \oplus_k 1$. In other words, every agent in
$I \times \{ t \oplus_k 1 \}$ appears in~$P_\alpha$ at most once, and every
element in~$P_\alpha$ belongs to $I \times \{ t \oplus_k 1 \}$. For every
$\alpha, \alpha', \alpha'' \in A$, we say that~$\alpha$ prefers~$\alpha'$
to~$\alpha''$ if~$\alpha'$ appears in~$P_\alpha$ and either agent~$\alpha''$
appears in~$P_\alpha$ after~$\alpha'$ or agent~$\alpha''$ does not appear
in~$P_\alpha$. We denote this $\DSMICYC{k}$ instance as
$X = (A, \{ P_\alpha \}_{\alpha \in A} )$.

Given a~$\DSMICYC{k}$ instance $X = (A, \{ P_\alpha \}_{\alpha \in A} )$,
a~\emph{family} is a~tuple
\begin{equation*}
(\alpha_0, \dots, \alpha_{k-1}) \in A^k
\end{equation*}
such that~$\alpha_t \in I \times \{t\}$ and $\alpha_{t \oplus_k 1}$ appears
in~$P_{\alpha_t}$ for every $t \in \{0, \dots, k-1\}$. A~\emph{matching}~$\mu$
is a~set of agent-disjoint families. In other words, for every
$(\alpha_0, \dots, \alpha_{k-1}), (\alpha'_0, \dots, \alpha'_{k-1}) \in \mu$
and $t, t' \in \{0, \dots, k-1\}$, if $\alpha_t = \alpha'_t$, then
$\alpha_{t'} = \alpha'_{t'}$. Given a~matching~$\mu$ and an agent
$\alpha \in A$, if $\alpha = \alpha_t$ for some
$(\alpha_0, \dots, \alpha_{k-1}) \in \mu$ and $t \in \{0, \dots, k-1\}$,
we say that~$\alpha$ is matched to~$\alpha_{t \oplus_k 1}$, and we write
$\mu(\alpha) = \alpha_{t \oplus_k 1}$. Otherwise, we say that~$\alpha$ is
unmatched, and we write $\mu(\alpha) = \alpha$.

Given a~matching~$\mu$, we say that a~family $(\alpha_0, \dots, \alpha_{k-1})$
is \emph{strongly blocking} if~$\alpha_t$ prefers~$\alpha_{t \oplus_k 1}$ to
$\mu(\alpha_t)$ for every $t \in \{0, \dots, k-1\}$. A matching~$\mu$ is
\emph{weakly stable} if it does not admit any strongly blocking family.

The $k$-dimensional stable matching problem with cyclic preferences
($\DSMCYC{k}$) is defined as the special case of $\DSMICYC{k}$ in which every
agent in $I \times \{ t \oplus_k 1\}$ appears exactly once in~$P_\alpha$ for every
agent $\alpha = (i, t) \in A$.

Notice that when incomplete lists are allowed, the case of an unequal number of
agents of each type can be handled within our $\DSMICYC{k}$ model by padding
with dummy agents whose preference lists are empty. Hence, the results of
Bir\'{o} and McDermid~\cite{BM10} apply to our $\DSMICYC{3}$ model. When
preference lists are complete, we follow the literature and focus on the case
where each type has an equal number of agents. Our result shows that even when
restricted to the case of an equal number of agents of each type, a~given
$\DSMCYC{k}$ instance need not admit a~weakly stable matching, and determining
the existence of a~weakly stable matching is NP-complete.

\subsection{Polynomial-Time Verification}

Given a~matching~$\mu$ of a~$\DSMICYC{k}$ instance with~$n$ agents per type,
it is straightforward to determine whether~$\mu$ is weakly stable in
$O(n^k)$~time by checking that none of the $O(n^k)$ families is strongly
blocking. The following theorem shows that when~$k$ is large, there is a~more
efficient method to determine whether a~given matching is weakly stable.

\begin{theorem} \label{thm:certificate}
There exists a~$\poly(n, k)$-time algorithm to determine
whether a~given matching~$\mu$ is weakly stable for a~$\DSMICYC{k}$ instance,
where~$n$ is the number of agents per type.
\end{theorem}

\begin{proof}
Given~$\mu$, consider the directed graph~$G$ with vertex set~$A$ such that
there exists an edge from~$\alpha$ to~$\alpha'$ if and only if~$\alpha$
prefers~$\alpha'$ to~$\mu(\alpha)$. Then cycles in~$G$ of length~$k$
correspond to strongly blocking families of~$\mu$. Notice that if a~cycle in~$G$
contains an agent in $I \times \{t\}$, then it also contains
an agent in $I \times \{t \oplus_k 1\}$. Hence no cycle in~$G$ has length
less than~$k$. Thus determining whether~$\mu$ is weakly stable is equivalent
to determining whether the directed graph~$G$ has a~cycle of length at
most~$k$, which can be done in $\poly(n, k)$ time.
\end{proof}

\section{NP-Completeness of \texorpdfstring{$k$}{k}-DSMI-CYC}
\label{sec:kdsmicyc}

In this section, we show that for every $k \geq 3$, it is NP-complete to
determine whether a~given $\DSMICYC{k}$ instance has a~weakly stable matching.
We achieve this by reducing from the problem of determining whether
a~$\DSMICYC{3}$ instance has a~weakly stable matching.

\subsection{The Reduction}
\label{sec:kdsmicyc-reduction}

Let $k \geq 4$. Consider an input~$\DSMICYC{3}$ instance
  $X = (A, \{ P_\alpha \}_{\alpha \in A} )$ where $A = I \times \{0, 1, 2\}$.
Our reduction constructs a~$\DSMICYC{k}$ instance
  $\Xhat = (\Ahat, \{ \Phat_{\alphahat} \}_{\alphahat \in \Ahat} )$ as follows.
\begin{itemize}
\item
Let $\Ihat = I \times I$ and $\Ahat = I \times I \times \{0, \dots, k-1\}$. For
every agent $(i, t) \in A$, we call $\alphahat = (i, i, t) \in \Ahat$
the non-dummy agent corresponding to $(i, t)$. We call the agents
\begin{equation*}
  \{ (i, j, t) \in \Ahat \mid t \notin \{0, 1, 2\} \text{ or } i \neq j\}
\end{equation*}
dummy agents.
\item
For every agent $\alphahat = (i, j, t) \in \Ahat$, we construct the
preference list~$\Phat_{\alphahat}$ as follows.
\begin{itemize}
\item
If $0 \leq t \leq 1$ and $i = j$, we list in~$\Phat_{\alphahat}$ the agents
\begin{equation*}
  \{ (i', j', t') \in I \times I \times \{t + 1\} \mid
  i' = j' \text{ and } (i', t') \text{ is in } P_{(i, t)} \}
\end{equation*}
in the order in which the corresponding agent $(i', t')$ appears
in~$P_{(i, t)}$.
\item
If $t = 2$ and $i = j$, we list in~$\Phat_{\alphahat}$ the agents
\begin{equation*}
  \{ (i', j', t') \in I \times I \times \{3\} \mid
  i' = i \text{ and } (j', 0) \text{ is in } P_{(i, 2)}\}
\end{equation*}
in the order in which the corresponding agent $(j', 0)$ appears
in~$P_{(i, 2)}$.
\item
If $0 \leq t \leq 2$ and $i \neq j$, we define~$\Phat_{\alphahat}$ as the empty
list.
\item
If $3 \leq t \leq k - 2$ and $(j, 0)$ is in~$P_{(i, 2)}$, we
define~$\Phat_{\alphahat}$ as $\List{(i, j, t + 1)}$.
\item
If $t = k - 1$ and $(j, 0)$ is in~$P_{(i, 2)}$, we
define~$\Phat_{\alphahat}$ as $\List{(j, j, 0)}$.
\item
If $3 \leq t \leq k - 1$ and $(j, 0)$ is not in~$P_{(i, 2)}$, we
define~$\Phat_{\alphahat}$ as the empty list.
\end{itemize}
\end{itemize}
Figure~\ref{fig:reduction} shows an example of the reduction when $k = 5$
and $I = \{0, 1\}$.

\subsection{Correctness of the Reduction}

\begin{lemma} \label{lem:correctness}
Let $k \geq 4$. Consider the reduction given in
Section~\ref{sec:kdsmicyc-reduction}. The output $\DSMICYC{k}$ instance~$\Xhat$
has a~weakly stable matching if and only if the input $\DSMICYC{3}$
instance~$X$ has a~weakly stable matching.
\end{lemma}

\begin{proof}
Since every family in~$\Xhat$ has the form
\begin{equation*}
( (i_0, i_0, 0), (i_1, i_1, 1), (i_2, i_2, 2),
  (i_2, i_0, 3), \dots, (i_2, i_0, k-1) ),
\end{equation*}
where $(i_0, i_1, i_2)$ is a~family in~$X$, there is a-one-to-one correspondence
between families in~$\Xhat$ and families in~$X$. This induces a~one-to-one
correspondence between matchings in~$\Xhat$ and matchings in~$X$. It is
straightforward to see that a~family in~$\Xhat$ is a~strongly blocking family
of a~matching in~$\Xhat$ if and only if the corresponding family in~$X$ is
a~strong blocking family of the corresponding matching in~$X$. Hence~$\Xhat$ has
a~weakly stable matching if and only if~$X$ has a~weakly stable matching.
\end{proof}

\begin{theorem} \label{thm:example}
Let $k \geq 3$. Then there exists a~$\DSMICYC{k}$ instance that has no weakly
stable matching.
\end{theorem}

\begin{proof}
Bir\'{o} and McDermid~\cite[Lemma~1]{BM10} show that there exists
a~$\DSMICYC{3}$ instance~$X$ that has no weakly stable matching. So we may
assume that $k \geq 4$. Then Lemma~\ref{lem:correctness} implies that given~$X$
as an input, the reduction in Section~\ref{sec:kdsmicyc-reduction} produces
a~$\DSMICYC{k}$ instance~$\Xhat$ that has no weakly stable matching.
\end{proof}

\begin{theorem} \label{thm:hardness}
Let $k \geq 3$. Then it is NP-complete to determine whether a~$\DSMICYC{k}$
instance has a~weakly stable matching.
\end{theorem}

\begin{proof}
Bir\'{o} and McDermid~\cite[Theorem~1]{BM10} show that it is NP-complete to
determine whether a~$\DSMICYC{3}$ instance has a~weakly stable matching. So we
may assume that $k \geq 4$. Then Lemma~\ref{lem:correctness} implies the
correctness of the reduction from $\DSMICYC{3}$ to $\DSMICYC{k}$ presented in
Section~\ref{sec:kdsmicyc-reduction}. Moreover, the reduction can be
implemented in $\poly(n, k)$ time, where~$n$ is the number of agents of each
type. Theorem~\ref{thm:certificate} implies that the problem of determining
whether a~$\DSMICYC{k}$ instance has a~weakly stable matching is in NP. Thus
the problem of determining whether a~$\DSMICYC{k}$ instance has a~weakly stable
matching is NP-complete.
\end{proof}

\begin{figure}[t]
\begin{center}
\subfloat[The input $\DSMICYC{3}$ instance.]{
  \includegraphics[scale=.9]{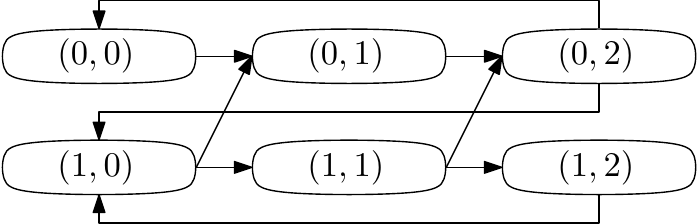}
}

\subfloat[The output $\DSMICYC{5}$ instance. Agents with empty preference lists
are omitted.]{
  \quad
  \includegraphics[scale=.9]{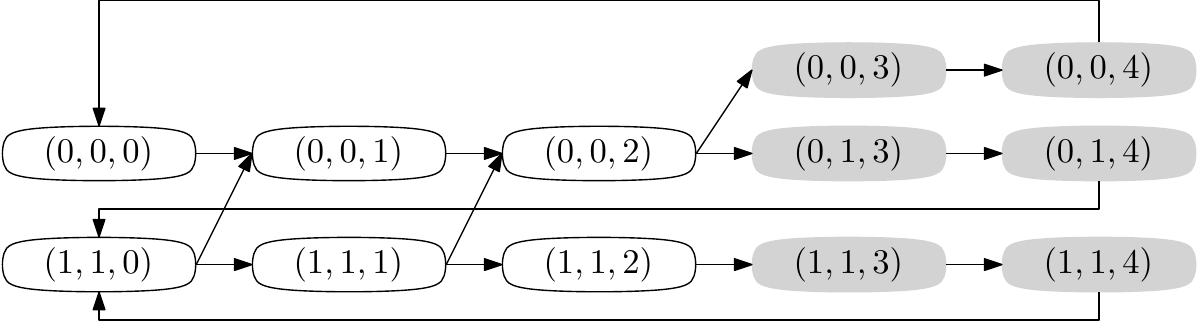}
  \quad
}
\end{center}
\caption{Example of a~reduction from $\DSMICYC{3}$ to $\DSMICYC{5}$. An arrow
indicates that the target agent appears in the preference list of the source
agent.}
\label{fig:reduction}
\end{figure}

\section{NP-Completeness of \texorpdfstring{$k$}{k}-DSM-CYC}
\label{sec:kdsmcyc}

In this section, we show that for every $k \geq 3$, it is NP-complete to
determine whether a~$\DSMCYC{k}$ instance has a~weakly stable matching. We
achieve this by reducing from the problem of determining whether
a~$\DSMICYC{k}$ instance has a~weakly stable matching. Since the dimensions of
both the input instance and the output instance of the reduction are equal
to~$k$, throughout this section, we write $\oplus$ instead of $\oplus_k$ for
better readability.

\subsection{The Reduction}
\label{sec:kdsmcyc-reduction}

Let $k \geq 3$. Consider an input $\DSMICYC{k}$ instance
$X = (A, \{ P_\alpha \}_{\alpha \in A} )$ where
$A = I \times \{0, \dots {k-1}\}$. We may assume that
$I = \{0, \dots, \abs{I} - 1\}$, so agents in~$A$ can be compared
lexicographically. Our reduction constructs a~$\DSMCYC{k}$ instance
  $\Xhat = (\Ahat, \{ \Phat_{\alphahat} \}_{\alphahat \in \Ahat} )$
as follows.

\begin{itemize}
\item
Let $J = \{0, \dots, (k - 1)^2\}$. Let $\Ihat = J \times A$ and
  $\Ahat = J \times A \times \{0, \dots, k - 1\}$.
For every agent $\alpha \in A$, we call
  $J \times \{\alpha\} \times \{0, \dots, k-1\}$
the gadget corresponding to~$\alpha$.
\item
For every agent $\alphahat = (j, \alpha, t) \in \Ahat$ such that $j = 0$ and
$\alpha \in I \times \{t\}$, we call~$\alphahat$ the non-dummy agent
corresponding to~$\alpha$. Let~$\Phat'_{\alpha}$ be the list obtained by
replacing every~$\alpha'$ in~$P_\alpha$ by~$(0, \alpha', t \oplus 1)$.
We define the preference list~$\Phat_{\alphahat}$ as
$\Phat'_{\alpha} \cdot
  \List{ (j', \alpha', t') \in J \times A \times \{t \oplus 1\} \mid
  \alpha' = \alpha }$
followed by the remaining agents in~$J \times A \times \{t \oplus 1\}$
in an arbitrary order.
\item
For every agent $\alphahat = (j, \alpha, t) \in \Ahat$ such that
$j = (k - 1)^2$, we call~$\alphahat$ a~boundary dummy agent, and we define
the preference list~$\Phat_{\alphahat}$ as
\begin{multline*}
\List{ (j', \alpha', t') \in J \times A \times \{t \oplus 1\} \mid
  \alpha' = \alpha \text{ and } j' < (k - 1)^2 } \\
\cdot \List{ (j', \alpha', t') \in J \times A \times \{t \oplus 1\} \mid
  j' = (k - 1)^2 }
\end{multline*}
followed by the remaining agents in~$J \times A \times \{t \oplus 1\}$
in an arbitrary order.
\item
For every agent $\alphahat = (j, \alpha, t) \in \Ahat$ such that
$(j, \alpha, t) \notin \{0\} \times (I \times \{t\}) \times \{t\}$ and
$j < (k-1)^2$, we call~$\alphahat$ a non-boundary dummy agent, and we
define the preference list~$\Phat_{\alphahat}$ as
$\List{ (j', \alpha', t') \in J \times A \times \{t \oplus 1\} \mid
  \alpha' = \alpha }$
followed by the remaining agents in~$J \times A \times \{t \oplus 1\}$
in an arbitrary order.
\end{itemize}
As shown in Figure~\subref*{fig:gadget-structure}, the gadget corresponding
to~$\alpha \in I \times \{t\}$ can be visualized as a~grid of agents with
$k$~rows and $(k-1)^2$~columns. The non-boundary dummy agents in the same
row have essentially the same preferences, which begin with the agents in the
next row from left to right. The preferences of the boundary dummy agents are
similar to those of the non-boundary dummy agents, except that they incorporate
the other boundary dummy agents in a~special manner. Meanwhile, the preferences
of the non-dummy agent~$(0, \alpha, t)$ reflect the preferences of
agent~$\alpha$ by starting with~$\Phat'_\alpha$.

\begin{figure}[t]
\begin{center}
\subfloat[The structure of the gadget.\label{fig:gadget-structure}]{
  \includegraphics[page=1,scale=.9]{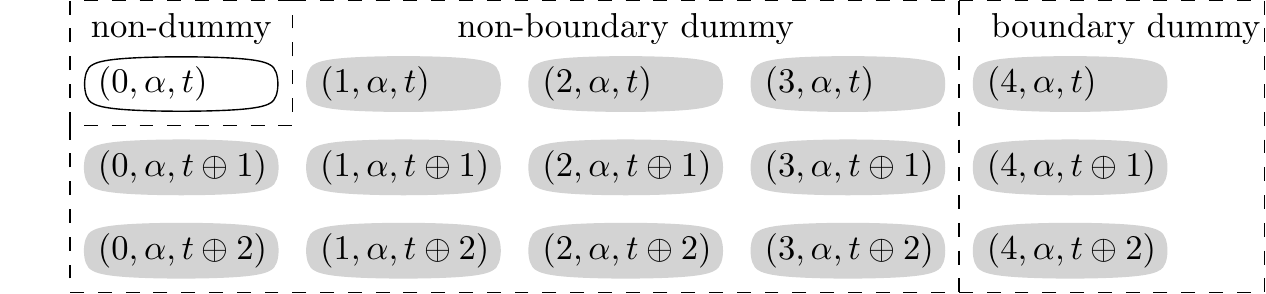}
}

\subfloat[The matching~$\muhat$ induced by~$\mu$ when~$\alpha$ is unmatched
in~$\mu$.\label{fig:gadget-singleton}]{
  \includegraphics[page=2,scale=.9]{gadget}
}

\subfloat[The matching~$\muhat$ induced by~$\mu$ when~$\alpha$ is matched
in~$\mu$.\label{fig:gadget-family}]{
  \includegraphics[page=3,scale=.9]{gadget}
}
\end{center}
\caption{Example of a~gadget corresponding to~$\alpha \in I \times \{t\}$ when
$k = 3$. An arrow indicates that the source agent is matched to the target
agent.}
\end{figure}

\subsection{Correctness of the Reduction}
Lemmas~\ref{lem:instability} and~\ref{lem:stability} below show that the
reduction in Section~\ref{sec:kdsmcyc-reduction} is a~correct reduction from
$\DSMICYC{k}$ to $\DSMCYC{k}$. The associated proofs are presented in
Sections~\ref{sec:kdsmcyc-instability} and~\ref{sec:kdsmcyc-stability}.

\begin{lemma} \label{lem:instability}
Let $k \geq 3$. Consider the reduction given in
Section~\ref{sec:kdsmcyc-reduction}. If the input $\DSMICYC{k}$ instance~$X$
has no weakly stable matching, then the output $\DSMCYC{k}$ instance~$\Xhat$
has no weakly stable matching.
\end{lemma}

\begin{lemma} \label{lem:stability}
Let $k \geq 3$. Consider the reduction in Section~\ref{sec:kdsmcyc-reduction}.
If the input $\DSMICYC{k}$ instance~$X$ has a weakly stable matching, then the
output $\DSMCYC{k}$ instance~$\Xhat$ has a weakly stable matching.
\end{lemma}

\begin{theorem}
Let $k \geq 3$. Then there exists a~$\DSMCYC{k}$ instance that has no weakly
stable matching.
\end{theorem}

\begin{proof}
By Theorem~\ref{thm:example}, there exists a~$\DSMICYC{k}$ instance~$X$ that
has no weakly stable matching. Then Lemma~\ref{lem:instability} implies that
given~$X$ as an input, the reduction in Section~\ref{sec:kdsmcyc-reduction}
produces a~$\DSMCYC{k}$ instance~$\Xhat$ that has no weakly stable matching.
\end{proof}

\begin{theorem}
Let $k \geq 3$. Then it is NP-complete to determine whether a~$\DSMCYC{k}$
instance has a~weakly stable matching.
\end{theorem}

\begin{proof}
By Theorem~\ref{thm:hardness}, it is NP-complete to determine whether
a~$\DSMICYC{k}$ instance has a~weakly stable matching.
Lemmas~\ref{lem:instability} and~\ref{lem:stability} imply the correctness of
the reduction from $\DSMICYC{k}$ to $\DSMCYC{k}$ presented in
Section~\ref{sec:kdsmcyc-reduction}. Moreover, the reduction can be implemented
in $\poly(n, k)$ time, where~$n$ is the number of agents of each type.
Theorem~\ref{thm:certificate} implies that the problem of determining whether
a~$\DSMCYC{k}$ instance has a~weakly stable matching is in NP. Thus the problem
of determining whether a~$\DSMCYC{k}$ instance has a~weakly stable matching is
NP-complete.
\end{proof}

\subsection{Properties of the Gadget}
\label{sec:kdsmcyc-gadget}

In this subsection, we study the properties of the gadget in the scenario that
the non-dummy agent is not matched to a~non-dummy agent corresponding to
an acceptable partner. In Lemma~\ref{lem:admirer}, we show that in this
scenario, many agents in the gadget are matched to agents in the same gadget.
In Lemma~\ref{lem:gadget}, we apply Lemma~\ref{lem:admirer} inductively to
show that in the same scenario, every agent in the same family as the non-dummy
agent belongs to the same gadget.

\begin{lemma} \label{lem:admirer}
Let $\muhat$ be a~weakly stable matching in~$\Xhat$. Let $t^* \in
\{0, \dots, k-1\}$ and $\alpha^* \in I \times \{t^*\}$ such that
$\muhat(0, \alpha^*, t^*)$ is not in~$\Phat'_{\alpha^*}$. Let
$t \in \{0, \dots, k-1\}$ and $j \in J$ such that $j \leq (k-1) \cdot (k-2)$.
Then $\muhat(j, \alpha^*, t) \in \{ 0, \dots, j + k - 1\} \times \{\alpha^*\}
\times \{t \oplus 1\}$.
\end{lemma}

\begin{proof}
Let $\Ahat_s = \List{(j', \alpha', t') \in J \times \{\alpha^*\} \times
\{t \oplus s \oplus 1\} \mid j' \leq j + k - s - 1}$ for every
$s \in \{0, \dots, k-1\}$. For the sake of contradiction, suppose
$\muhat(j, \alpha^*, t)$ is not in $\Ahat_0$.

For every $s \in \{0, \dots, k-2\}$, since the length of $\Ahat_s$ is greater
than the length of $\Ahat_{s+1}$, there exists~$\alphahat_s$ in~$\Ahat_s$ such
that $\muhat(\alphahat_s)$ is not in $\Ahat_{s+1}$. Let
$\alphahat_{k-1} = (j, \alpha^*, t)$. Then $\alphahat_{k-1}$ is
in~$\Ahat_{k-1}$ and~$\muhat(\alphahat_{k-1})$ is not in~$\Ahat_0$.
Since~$\muhat$ is a~weakly stable matching of~$\Xhat$, the family
$(\alphahat_{k - t - 1}, \dots, \alphahat_{(k - t - 1) \oplus (k-1)})$ is not
strongly blocking. So there exists $s^* \in \{0, \dots, k-1\}$ such
that~$\alphahat_{s^*}$ does not prefer~$\alphahat_{s^* \oplus 1}$ to
$\mu(\alphahat_{s^*})$. Since~$\alphahat_{s^*}$ is in~$\Ahat_{s^*}$, there
exists $j^* \leq j + k - s^* - 1$ such that
$\alphahat_{s^*} = (j^*, \alpha^*, t \oplus s^* \oplus 1)$. We consider two
cases.

Case~1: $j^* = 0$ and $t \oplus s^* \oplus 1 = t^*$.
Then $\alphahat_{s^*} = (0, \alpha^*, t^*)$ is a non-dummy agent and
$\Phat'_{\alpha^*} \cdot \Ahat_{s^* \oplus 1}$ is a prefix of the
preference list~$\Phat_{\alphahat_{s^*}}$. Since $\mu(\alphahat_{s^*})$ is not
in $\Phat'_{\alpha^*} \cdot \Ahat_{s^* \oplus 1}$ and
$\alphahat_{s^* \oplus 1}$ is in~$\Ahat_{s^* \oplus 1}$, agent
$\alphahat_{s^*}$ prefers~$\alphahat_{s^* \oplus 1}$ to $\mu(\alphahat_{s^*})$,
a contradiction.

Case~2: $j^* \neq 0$ or $t \oplus s^* \oplus 1 \neq t^*$. We consider two
subcases.

Case~2.1: $j^* = (k-1)^2$.
Since $(k-1)^2 = j^* \leq j + k - s^* - 1 \leq (k-1)^2 - s^*$, we have
$s^* = 0$. Hence $\alphahat_0 = ((k-1)^2, \alpha^*, t \oplus 1)$ is a
boundary dummy agent and~$\Ahat_1$ is a prefix of the preference
list~$\Phat_{\alphahat_0}$. Since $\mu(\alphahat_0)$ is not
in~$\Ahat_1$ and $\alphahat_1$ is in~$\Ahat_1$, agent
$\alphahat_0$ prefers~$\alphahat_1$ to $\mu(\alphahat_{0})$,
a contradiction.

Case~2.2: $j^* < (k-1)^2$.
Then~$\alphahat_{s^*}$ is a non-boundary dummy agent
and~$\Ahat_{s^* \oplus 1}$ is a prefix of the preference
list~$\Phat_{\alphahat_{s^*}}$. Since $\mu(\alphahat_{s^*})$ is not
in~$\Ahat_{s^* \oplus 1}$ and $\alphahat_{s^* \oplus 1}$ is
in~$\Ahat_{s^* \oplus 1}$, agent $\alphahat_{s^*}$
prefers~$\alphahat_{s^* \oplus 1}$ to $\mu(\alphahat_{s^*})$, a contradiction.
\end{proof}

\begin{remark}
\label{rem:admirer}
In the proof of Lemma~\ref{lem:admirer}, we can think of $\alphahat_0, \dots,
\alphahat_{k-1}$ as a chain of admirers in the gadget corresponding
to~$\alpha$, where~$\alphahat_s$ prefers~$\alphahat_{s + 1}$ to
$\muhat(\alphahat_s)$. By applying the weak stability condition to this chain of
admirers, we show that~$\alphahat_{k-1}$ is matched to a~partner no worse than
$\alphahat_0$.
\end{remark}

\begin{lemma} \label{lem:gadget}
Let~$\muhat$ be a~weakly stable matching in~$\Xhat$. Let
$j_0, \dots, j_{k-1} \in J$ and $\alpha_0, \dots, \alpha_{k-1} \in A$
such that $( (j_0, \alpha_0, 0), \dots, (j_{k-1}, \alpha_{k-1}, k-1) )
\in \muhat$. Let $t^* \in \{0, \dots, k-1\}$ such that $j_{t^*} = 0$
and $\alpha_{t^*} \in I \times \{t^*\}$. Suppose that
$(j_{t^* \oplus 1}, \alpha_{t^* \oplus 1}, t^* \oplus 1)$
is not in~$\Phat'_{\alpha_{t^*}}$. Then, for every $s \in \{0, \dots, k-1\}$,
we have~$\alpha_{t^* \oplus s} = \alpha_{t^*}$ and
$j_{t^* \oplus s} \leq (k-1) \cdot s$.
\end{lemma}

\begin{proof}
We prove the claim by induction on~$s$. When~$s = 0$, we have
$\alpha_{t^* \oplus s} = \alpha_{t^* \oplus 0} = \alpha_{t^*}$ and
$j_{t^* \oplus s} = j_{t^*} = 0 \leq (k-1) \cdot s$.

Suppose $\alpha_{t^* \oplus (s - 1)} = \alpha_{t^*}$ and
$j_{t^* \oplus (s - 1)} \leq (k-1) \cdot (s - 1)$, where
$s \in \{1, \dots, k-1\}$. Since
$(j_{t^* \oplus 1}, \alpha_{t^* \oplus 1}, t^* \oplus 1)$
is not in~$\Phat'_{\alpha_{t^*}}$, agent~$\muhat(0, \alpha_{t^*}, t^*)$
is not in~$\Phat'_{\alpha_{t^*}}$. Let $t = t^* \oplus (s - 1)$. Then
$\alpha_t = \alpha_{t^* \oplus (s - 1)} = \alpha_{t^*}$ and
$j_t = j_{t^* \oplus (s - 1)} \leq (k-1) \cdot (s - 1) \leq (k-1) \cdot (k-2)$.
So Lemma~\ref{lem:admirer} implies that
$\muhat(j_t, \alpha_{t^*}, t)
\in \{0, \dots, j_t + k - 1\} \times \{\alpha_{t^*}\} \times \{t \oplus 1\}$.
Hence $j_{t \oplus 1} \leq j_t + k - 1$ and
$\alpha_{t \oplus 1} = \alpha_{t^*}$, since~$\muhat(j_t, \alpha_{t^*}, t)
= \muhat(j_t, \alpha_t, t)
= (j_{t \oplus 1}, \alpha_{t \oplus 1}, t \oplus 1)$. Thus
$\alpha_{t^* \oplus s} = \alpha_{t \oplus 1} = \alpha_{t^*}$ and
$j_{t^* \oplus s} = j_{t \oplus 1} \leq j_t + k - 1
= j_{t^* \oplus (s-1)} + k - 1
\leq (k-1) \cdot (s-1) + k - 1 = (k-1) \cdot s$.
\end{proof}

\subsection{Proof of Lemma~\ref{lem:instability}}
\label{sec:kdsmcyc-instability}

The goal of this subsection is to prove Lemma~\ref{lem:instability}. It suffices
to show that every weakly stable matching~$\muhat$ in~$\Xhat$ induces a weakly
stable matching~$\mu$ in~$X$.

Recall that each agent in~$A$ has a~corresponding non-dummy agent in~$\Ahat$,
and that a~family in~$X$ is a~tuple of $k$~agents in~$A$ such that each agent
appears in the preference list of another. Hence we include in~$\mu$ a~family
of agents in~$X$ whenever the corresponding family of non-dummy agents are
matched in~$\muhat$. More formally, we define the matching~$\mu$ in~$X$ induced
by~$\muhat$ in~$\Xhat$ as the set of families
$(\alpha_0, \dots, \alpha_{k-1})$ in~$X$ satisfying
  $((0, \alpha_0, 0), \dots, (0, \alpha_{k-1}, k-1)) \in \muhat$.
Notice that every~$\mu$ induced by a~matching~$\muhat$ in~$\Xhat$ is a~valid
matching in~$X$ since agent-disjoint families in~$\Xhat$ induce agent-disjoint
families in~$X$.

Lemma~\ref{lem:correspondence} below shows that if $\muhat$ is weakly stable and
matches a~non-dummy agent to a~non-dummy agent corresponding to an
acceptable partner, then~$\mu$ matches the corresponding agents. Our proof
relies on Lemma~\ref{lem:gadget} and the weak stability of~$\muhat$. Notice
that if~$\muhat$ is not weakly stable, it may be the case that~$\muhat$ matches
a~family consisting of $k-1$~non-dummy agents and one dummy agent. In such
a~case, the corresponding $k-1$ agents are unmatched in the induced
matching~$\mu$.

\begin{lemma} \label{lem:correspondence}
Let~$\mu$ be the matching in~$X$ induced by a~weakly stable matching~$\muhat$
in~$\Xhat$. Let $t \in \{0, \dots, k-1\}$ and $\alpha \in I \times \{t\}$
such that $\muhat(0, \alpha, t)$ is in~$\Phat'_{\alpha}$. Then
$\muhat(0, \alpha, t) = (0, \mu(\alpha), t \oplus 1)$.
\end{lemma}

\begin{proof}
For the sake of contradiction, suppose~$\muhat(0, \alpha, t)
\neq (0, \mu(\alpha), t \oplus 1)$. Since $\muhat(0, \alpha, t)$ is
in~$\Phat'_{\alpha}$, we have $( (j_0, \alpha_0, 0), \dots, (j_{k-1},
\alpha_{k-1}, k-1) ) \in \muhat$ for some $j_0, \dots, j_{k-1} \in J$
and $\alpha_0, \dots, \alpha_{k-1} \in A$ such that $(j_t, \alpha_t, t)
= (0, \alpha, t)$ and $(j_{t \oplus 1}, \alpha_{t \oplus 1}, t \oplus 1)$
is in~$\Phat'_{\alpha}$. Let
\begin{equation*}
T = \{ t' \in \{0, \dots, k-1\} \mid
\alpha_{t'} \in I \times \{t'\} \text{ and }
(j_{t' \oplus 1}, \alpha_{t' \oplus 1}, t' \oplus 1) \text{ is in }
\Phat'_{\alpha_{t'}} \}.
\end{equation*}
Then $t \in T$. We consider two cases.

Case~1: $T = \{0, \dots, k-1\}$.
Then for every $t' \in T = \{0, \dots, k-1\}$, we have
$\alpha_{t'} \in I \times \{t'\}$ and
$(j_{t' \oplus 1}, \alpha_{t' \oplus 1}, t' \oplus 1)$ is
in~$\Phat'_{\alpha_{t'}}$. So $j_{t' \oplus 1} = 0$ and
$\alpha_{t' \oplus 1}$ is in~$P_{\alpha_{t'}}$ for every $t' \in
\{0, \dots, k-1\}$. Hence $(\alpha_0, \dots, \alpha_{k-1})$ is
a~valid family in~$X$. Since~$\mu$ is induced by~$\muhat$ and
$( (0, \alpha_0, 0), \dots, (0, \alpha_{k-1}, k-1) ) \in \muhat$,
we have $(\alpha_0, \dots, \alpha_{k-1}) \in \mu$. Thus~$\mu(\alpha)
= \mu(\alpha_t) = \alpha_{t \oplus 1}$, which contradicts
$(0, \mu(\alpha), t \oplus 1) \neq \muhat(0, \alpha, t)
= (0, \alpha_{t \oplus 1}, t \oplus 1)$.

Case~2: $T \neq \{0, \dots, k-1\}$.
Then there exists a~smallest $s^* \in \{1, \dots, k-1\}$ such that
$t \oplus s^* \notin T$. Then $t \oplus (s^* - 1) \in T$. Let
$t^* = t \oplus s^*$. Since $t^* \oplus (-1) = t \oplus (s^* - 1) \in T$, we
have $\alpha_{t^* \oplus (-1)} \in I \times \{ t^* \oplus (-1) \}$
and $(j_{t^*}, \alpha_{t^*}, t^*)$ is in~$\Phat'_{\alpha_{t^* \oplus (-1)}}$.
So $j_{t^*} = 0$ and~$\alpha_{t^*}$ is in~$P_{\alpha_{t^* \oplus (-1)}}$. Hence
$\alpha_{t^*} \in I \times \{t^*\}$. Since $\alpha_{t^*} \in I \times \{t^*\}$
and $t^* = t \oplus s^* \notin T$, agent
$(j_{t^* \oplus 1}, \alpha_{t^* \oplus 1}, t^* \oplus 1)$ is not
in~$\Phat'_{\alpha_{t^*}}$. So Lemma~\ref{lem:gadget} implies
$\alpha_{t^* \oplus (k-1)} = \alpha_{t^*}$. Hence
$\alpha_{t^* \oplus (-1)} = \alpha_{t^* \oplus (k-1)} = \alpha_{t^*}
\in I \times \{t^*\}$, which contradicts $\alpha_{t^* \oplus (-1)} \in
I \times \{t^* \oplus (-1)\}$.
\end{proof}

\begin{proof}[Proof of Lemma~\ref{lem:instability}]
For the sake of contradiction, suppose~$X$ has no weakly stable matching and
$\Xhat$ has a~weakly stable matching~$\muhat$. Let~$\mu$ be the matching in~$X$
induced by~$\muhat$.

Since~$\mu$ is not a~weakly stable matching of~$X$, there exists a~strongly
blocking family $(\alpha_0, \dots, \alpha_{k-1})$. Since~$\muhat$ is a~weakly
stable matching of~$\Xhat$, the family
\begin{equation*}
( (0, \alpha_0, 0), \dots, (0, \alpha_{k-1}, k-1) )
\end{equation*}
is not strongly blocking. So there exists $t \in \{0, \dots, k-1\}$ such that
$(0, \alpha_t, t)$ does not prefer $(0, \alpha_{t \oplus 1}, t \oplus 1)$ to
$\muhat(0, \alpha_t, t)$. Since $(\alpha_0, \dots, \alpha_{k-1})$ is a~family
in~$X$, agent~$\alpha_{t \oplus 1}$ is in~$P_{\alpha_t}$. So
$(0, \alpha_{t \oplus 1}, t \oplus 1)$ is in~$\Phat'_{\alpha_t}$. Hence
$\muhat(0, \alpha_t, t)$ appears in~$\Phat'_{\alpha_t}$ no later
than~$(0, \alpha_{t \oplus 1}, t \oplus 1)$, since~$\Phat'_{\alpha_t}$
is a~prefix of the preference list~$\Phat_{(0, \alpha_t, t)}$.

Since $\muhat(0, \alpha_t, t)$ is in~$\Phat'_{\alpha_t}$,
Lemma~\ref{lem:correspondence} implies $\muhat(0, \alpha_t, t) =
(0, \mu(\alpha_t), t \oplus 1)$. Since $(0, \mu(\alpha_t), t \oplus 1)$
appears in~$\Phat'_{\alpha_t}$ no later
than~$(0, \alpha_{t \oplus 1}, t \oplus 1)$, agent $\mu(\alpha_t)$
appears in~$P_{\alpha_t}$ no later than~$\alpha_{t \oplus 1}$. Hence~$\alpha_t$
does not prefer~$\alpha_{t \oplus 1}$ to $\mu(\alpha_t)$. So
$(\alpha_0, \dots, \alpha_{k-1})$ is not a~strongly blocking family of~$\mu$,
a~contradiction.
\end{proof}

\subsection{Proof of Lemma~\ref{lem:stability}}
\label{sec:kdsmcyc-stability}

The goal of this subsection is to prove Lemma~\ref{lem:stability}. It suffices
to show that every weakly stable matching~$\mu$ in~$X$ induces a weakly stable
matching~$\muhat$ in~$\Xhat$.
We construct the matching~$\muhat$ induced by~$\mu$ as follows.

\begin{itemize}
\item
For every $(\alpha_0, \dots, \alpha_{k-1}) \in \mu$, we include
in~$\muhat$ the family
\begin{equation*}
( (0, \alpha_0, 0), \dots, (0, \alpha_{k-1}, k-1)).
\end{equation*}
\item
For every agent $\alpha \in A$ and $j \in J$ such that $j < (k-1)^2$, we include
in~$\muhat$ the family $( (j + \delta_0(\alpha), \alpha, 0), \dots,
(j + \delta_{k-1}(\alpha), \alpha, k-1))$, where
\begin{equation*}
  \delta_t(\alpha) =
  \begin{cases}
    1 & \text{if } \mu(\alpha) \neq \alpha \text{ and } \alpha \in I \times \{t\} \\
    0 & \text{otherwise}
  \end{cases}
\end{equation*}
\item
For every $t \in \{0, \dots, k-1\}$, let $R_t$ be the list
\begin{equation*}
  \List{(j', \alpha', t') \in \{ (k-1)^2 \} \times A \times \{t\} \mid
  \delta_{t'}(\alpha') = 0 }.
\end{equation*}
We include in~$\muhat$ the family~$(R_0[s], \dots, R_{k-1}[s])$ for every
$0 \leq s < \abs{A} - \abs{\mu}$, where $R_t[s]$ denotes the $(s + 1)$th
element of~$R_t$.
\end{itemize}
Figures~\subref*{fig:gadget-singleton} and~\subref*{fig:gadget-family} show the
gadget under the matching~$\muhat$.

It is straightforward to check that the families in~$\muhat$ induced by
a~matching~$\mu$ are agent-disjoint. Hence~$\muhat$ is a~valid matching
in~$\Xhat$.

\begin{lemma} \label{lem:preference}
Let~$\muhat$ be the matching in~$\Xhat$ induced by a~matching~$\mu$ in~$X$. Let
$t \in \{0, \dots, k-1\}$ and $\alpha \in A$ such that
$\alpha \in I \times \{t\}$. Let $j' \in J$ and $\alpha' \in A$ such that
non-dummy agent $(0, \alpha, t)$ prefers $(j', \alpha', t \oplus 1)$ to
$\muhat(0, \alpha, t)$. Then $(j', \alpha', t \oplus 1)$ is in $\Phat'_{\alpha}$
and~$\alpha$ prefers~$\alpha'$ to~$\mu(\alpha)$.
\end{lemma}

\begin{proof}
Notice that $\Phat'_{\alpha} \cdot \List{ (0, \alpha, t \oplus 1) }$ is
a~prefix of the preference list~$\Phat_{(0, \alpha, t)}$ of non-dummy agent
$(0, \alpha, t)$. We consider two cases.

Case~1: $\mu(\alpha) \neq \alpha$.
Then $\muhat(0, \alpha, t) = (0, \mu(\alpha), t \oplus 1)$. Since
$(0, \alpha, t)$ prefers $(j', \alpha', t \oplus 1)$ to
$(0, \mu(\alpha), t \oplus 1)$, agent $(j', \alpha', t \oplus 1)$
appears in $\Phat'_{\alpha}$ before $(0, \mu(\alpha), t \oplus 1)$.
Hence~$\alpha$ prefers~$\alpha'$ to~$\mu(\alpha)$.

Case 2: $\mu(\alpha) = \alpha$.
Then $\muhat(0, \alpha, t) = (0, \alpha, t \oplus 1)$. Since
$(0, \alpha, t)$ prefers $(j', \alpha', t \oplus 1)$ to
${(0, \alpha, t \oplus 1)}$, agent $(j', \alpha', t \oplus 1)$ is
in~$\Phat'_{\alpha}$. Then~$\alpha'$ is in~$P_{\alpha}$, and hence~$\alpha$
prefers~$\alpha'$ to~$\mu(\alpha)$.
\end{proof}

\begin{lemma} \label{lem:boundary}
Let~$\muhat$ be the matching in~$\Xhat$ induced by a~weakly stable
matching~$\mu$ in~$X$. Let $j_0, \dots, j_{k-1} \in J$ and
$\alpha_0, \dots, \alpha_{k-1} \in A$ such that
\begin{equation*}
((j_0, \alpha_0, 0), \dots, (j_{k-1}, \alpha_{k-1}, k-1))
\end{equation*} is a~strongly blocking family of~$\muhat$. Then
$j_t - \delta_t(\alpha_t) \geq (k-1)^2$ for every $t \in \{0, \dots, k-1\}$.
\end{lemma}

\begin{proof}
Let $t^* \in \{0, \dots, k-1\}$ such that
\begin{equation*}
j_{t^*} - \delta_{t^*}(\alpha_{t^*})
= \min_{t \in \{0, \dots, k-1\}} (j_t - \delta_t(\alpha_t)).
\end{equation*}
For the sake of contradiction, suppose
$j_{t^*} - \delta_{t^*}(\alpha_{t^*}) < (k-1)^2$. We consider two cases.

Case~1: $j_{t^*} = 0$ and $\alpha_{t^*} \in I \times \{t^*\}$.
Let $T = \{ t \mid j_t = 0 \text{ and } \alpha_t \in I \times \{t\} \}$.
Then $t^* \in T$. We consider two subcases.

Case~1.1: $T = \{0, \dots, k-1\}$.
Then for every $t \in \{0, \dots, k-1\} = T$, since $(0, \alpha_t, t)$
prefers $(0, \alpha_{t \oplus 1}, t \oplus 1)$ to $\muhat(0, \alpha_t, t)$,
Lemma~\ref{lem:preference} implies that~$\alpha_t$
prefers~$\alpha_{t \oplus 1}$ to~$\mu(\alpha_t)$. Hence
$(\alpha_0, \dots, \alpha_{k-1})$ is a~strongly
blocking family of~$\mu$, which contradicts the stability of~$\mu$.

Case~1.2: $\{t^*\} \subseteq T \subsetneq \{0, \dots, k-1\}$.
Then there exists~$s^*$ such that $s^* \in T$ and $s^* \oplus 1 \notin T$.
Since $s^* \in T$, we have $j_{s^*} = 0$ and
$\alpha_{s^*} \in I \times \{s^*\}$. Since $(0, \alpha_{s^*}, s^*)$ prefers
$(j_{s^* \oplus 1}, \alpha_{s^* \oplus 1}, s^* \oplus 1)$ to
$\muhat(0, \alpha_{s^*}, s^*)$, Lemma~\ref{lem:preference} implies that
$(j_{s^* \oplus 1}, \alpha_{s^* \oplus 1}, s^* \oplus 1)$ is
in~$\Phat'_{\alpha_{s^*}}$. Hence $j_{s^* \oplus 1} = 0$ and
$\alpha_{s^* \oplus 1} \in I \times \{s^* \oplus 1\}$, which contradicts
$s^* \oplus 1 \notin T$.

Case~2: Either $j_{t^*} \neq 0$ or $\alpha_{t^*} \notin I \times \{t^*\}$.
Thus $(j_{t^*}, \alpha_{t^*}, t^*)$ is a~dummy agent. We consider two subcases.

Case~2.1: $j_{t^*} < (k-1)^2$.
Since
\begin{equation*}
\muhat(j_{t^*}, \alpha_{t^*}, t^*)
= (j_{t^*} + \delta_{t^* \oplus 1}(\alpha_{t^*})
  - \delta_{t^*}(\alpha_{t^*}), \alpha_{t^*}, t^* \oplus 1),
\end{equation*}
and the non-boundary dummy agent $(j_{t^*}, \alpha_{t^*}, t^*)$ prefers
$(j_{t^* \oplus 1}, \alpha_{t^* \oplus 1}, t^* \oplus 1)$ to
$\muhat(j_{t^*}, \alpha_{t^*}, t^*)$, we have
$j_{t^* \oplus 1} < j_{t^*}
  + \delta_{t^* \oplus 1}(\alpha_{t^*}) - \delta_{t^*}(\alpha_{t^*})$, which
contradicts the definition of~$t^*$.

Case~2.2: $j_{t^*} = (k-1)^2$.
Then $\delta_{t^*}(\alpha_{t^*}) = 1$ since
$j_{t^*} - \delta_{t^*}(\alpha_{t^*}) < (k-1)^2$.
So $\alpha_{t^*} \in I \times \{t^*\}$, and hence
$\delta_{t^* \oplus 1}(\alpha_{t^*}) = 0$. Since
\begin{equation*}
\muhat(j_{t^*}, \alpha_{t^*}, t^*)
= (j_{t^*} - 1, \alpha_{t^*}, t^* \oplus 1)
\end{equation*}
and the boundary dummy agent $(j_{t^*}, \alpha_{t^*}, t^*)$ prefers
$(j_{t^* \oplus 1}, \alpha_{t^* \oplus 1}, t^* \oplus 1)$ to
$\muhat(j_{t^*}, \alpha_{t^*}, t^*)$, we have
$j_{t^* \oplus 1} < j_{t^*} - 1 = j_{t^*}
  + \delta_{t^* \oplus 1}(\alpha_{t^*}) - \delta_{t^*}(\alpha_{t^*})$, which
contradicts the definition of~$t^*$.
\end{proof}

\begin{proof}[Proof of Lemma~\ref{lem:stability}]
Suppose~$X$ has a~weakly stable matching~$\mu$. Let~$\muhat$ be the matching
in~$\Xhat$ induced by~$\mu$. It suffices to show that~$\muhat$ does not admit
a~strongly blocking family.

For the sake of contradiction, suppose~$\muhat$ admits a~strongly blocking
family
\begin{equation*}
((j_0, \alpha_0, 0), \dots, (j_{k-1}, \alpha_{k-1}, k-1)).
\end{equation*}
Lemma~\ref{lem:boundary} implies that for every $t \in \{0, \dots, k-1\}$, we
have $j_t - \delta_t(\alpha_t) \geq (k-1)^2$. Since $j_t \leq (k-1)^2$ and
$\delta_t(\alpha_t) \geq 0$, we deduce that $j_t = (k-1)^2$ and
$\delta_t(\alpha_t) = 0$ for every $t \in \{0, \dots, k-1\}$. Hence for every
$t \in \{0, \dots, k-1\}$, there exists~$s_t$ such that
$(j_t, \alpha_t, t) = R_t[s_t]$.

Let $t^* \in \{0, \dots, k-1\}$ such that
\begin{equation*}
s_{t^*} = \min_{t \in \{0, \dots, k-1\}} s_t.
\end{equation*}
Since $\muhat(R_{t^*}[s_{t^*}]) = R_{t^* \oplus 1}[s_{t^*}]$ and
the boundary dummy agent $R_{t^*}[s_{t^*}]$ prefers boundary dummy agent
$R_{t^* \oplus 1}[s_{t^* \oplus 1}]$ to boundary dummy agent
$\muhat(R_{t^*}[s_{t^*}])$, we deduce that
$R_{t^* \oplus 1}[s_{t^* \oplus 1}]$ is lexicographically smaller than
$R_{t^* \oplus 1}[s_{t^*}]$. Hence $s_{t^* \oplus 1} < s_{t^*}$, which
contradicts the definition of~$t^*$.
\end{proof}

\section{Concluding Remarks}
\label{sec:conclusion}

We have shown that a~$\DSMCYC{3}$ instance need not admit a~weakly stable
matching, and that it is NP-complete to determine whether a~given $\DSMCYC{3}$
instance admits a~weakly stable matching. It seems that for the
three-dimensional stable matching problem, none of the preference structures
studied in the literature admits a~non-trivial generalization of the existence
theorem of Gale and Shapley. (The existence result in Danilov's
model~\cite{Dan03} follows from applying the Gale-Shapley algorithm in
a~straightforward manner.) It would be interesting to consider solution
concepts such as popular matchings instead of stable matchings in the
multi-dimensional matching context.

\bibliographystyle{plainnat}
\bibliography{kdsmcyc}

\begin{thebibliography}{18}
\providecommand{\natexlab}[1]{#1}
\providecommand{\url}[1]{\texttt{#1}}
\expandafter\ifx\csname urlstyle\endcsname\relax
  \providecommand{\doi}[1]{doi: #1}\else
  \providecommand{\doi}{doi: \begingroup \urlstyle{rm}\Url}\fi

\bibitem[Alkan(1988)]{Alk88}
A.~Alkan.
\newblock Nonexistence of stable threesome matchings.
\newblock \emph{Mathematical Social Sciences}, 16\penalty0 (2):\penalty0
  207--209, 1988.

\bibitem[Bir\'{o} and McDermid(2010)]{BM10}
P.~Bir\'{o} and E.~McDermid.
\newblock Three-sided stable matchings with cyclic preferences.
\newblock \emph{Algorithmica}, 58\penalty0 (1):\penalty0 5--18, 2010.

\bibitem[Boros et~al.(2004)Boros, Gurvich, Jaslar, and Krasner]{BGJK04}
E.~Boros, V.~Gurvich, S.~Jaslar, and D.~Krasner.
\newblock Stable matchings in three-sided systems with cyclic preferences.
\newblock \emph{Discrete Mathematics}, 289\penalty0 (1):\penalty0 1--10, 2004.

\bibitem[Cui and Jia(2013)]{CJ13}
L.~Cui and W.~Jia.
\newblock Cyclic stable matching for three-sided networking services.
\newblock \emph{Computer Networks}, 57\penalty0 (1):\penalty0 351--363, 2013.

\bibitem[Danilov(2003)]{Dan03}
V.~I. Danilov.
\newblock Existence of stable matchings in some three-sided systems.
\newblock \emph{Mathematical Social Sciences}, 46\penalty0 (2):\penalty0
  145--148, 2003.

\bibitem[Eriksson et~al.(2006)Eriksson, Sj\"{o}strand, and Strimling]{ESS06}
K.~Eriksson, J.~Sj\"{o}strand, and P.~Strimling.
\newblock Three-dimensional stable matching with cyclic preferences.
\newblock \emph{Mathematical Social Sciences}, 52\penalty0 (1):\penalty0
  77--87, 2006.

\bibitem[Escamocher and O'Sullivan(2018)]{EO18}
G.~Escamocher and B.~O'Sullivan.
\newblock Three-dimensional matching instances are rich in stable matchings.
\newblock In \emph{Proceedings of the 15th International Conference on the
  Integration of Constraint Programming, Artificial Intelligence, and
  Operations Research}, pages 182--197, 2018.

\bibitem[Farczadi et~al.(2016)Farczadi, Georgiou, and K\"{o}nemann]{FGK16}
L.~Farczadi, K.~Georgiou, and J.~K\"{o}nemann.
\newblock Stable marriage with general preferences.
\newblock \emph{Theory of Computing Systems}, 59\penalty0 (4):\penalty0
  683--699, 2016.

\bibitem[Gale and Shapley(1962)]{GS62}
D.~Gale and L.~S. Shapley.
\newblock College admissions and the stability of marriage.
\newblock \emph{American Mathematical Monthly}, 69\penalty0 (1):\penalty0
  9--15, 1962.

\bibitem[Hofbauer(2016)]{Hof16}
J.~Hofbauer.
\newblock $d$-dimensional stable matching with cyclic preferences.
\newblock \emph{Mathematical Social Sciences}, 82:\penalty0 72--76, 2016.

\bibitem[Huang(2007)]{Hua07}
C.-C. Huang.
\newblock Two's company, three's a~crowd: {S}table family and threesome
  roommates problems.
\newblock In \emph{Proceedings of the 15th Annual European Symposium on
  Algorithms}, pages 558--569, 2007.

\bibitem[Huang(2010)]{Hua10}
C.-C. Huang.
\newblock Circular stable matching and 3-way kidney transplant.
\newblock \emph{Algorithmica}, 58\penalty0 (1):\penalty0 137--150, 2010.

\bibitem[Knuth(1997)]{Knu97}
D.~E. Knuth.
\newblock \emph{Stable marriage and its relation to other combinatorial
  problems: An~introduction to the mathematical analysis of algorithms}.
\newblock American Mathematical Society, Providence, RI, 1997.

\bibitem[Manlove(2013)]{Man13}
D.~F. Manlove.
\newblock \emph{Algorithmics of matching under preferences}.
\newblock World Scientific, Singapore, 2013.

\bibitem[Ng and Hirschberg(1991)]{NH91}
C.~Ng and D.~Hirschberg.
\newblock Three-dimensional stable matching problems.
\newblock \emph{SIAM Journal on Discrete Mathematics}, 4\penalty0 (2):\penalty0
  245--252, 1991.

\bibitem[Pashkovich and Poirrier(2018)]{PP18}
K.~Pashkovich and L.~Poirrier.
\newblock Three-dimensional stable matching with cyclic preferences, 2018.
\newblock arXiv:1807.05638.

\bibitem[Subramanian(1994)]{Sub94}
A.~Subramanian.
\newblock A~new approach to stable matching problems.
\newblock \emph{SIAM Journal on Computing}, 23\penalty0 (4):\penalty0 671--700,
  1994.

\bibitem[Woeginger(2013)]{Woe13}
G.~J. Woeginger.
\newblock Core stability in hedonic coalition formation.
\newblock In \emph{Proceedings of the 39th International Conference on Current
  Trends in Theory and Practice of Computer Science}, pages 33--50, 2013.

\end{thebibliography}

\end{document}